%% file: main.tex
\documentclass[conference, letterpaper, 10pt]{IEEEtran}

\IEEEoverridecommandlockouts      
\usepackage{amsmath,amssymb,amsfonts,amsthm,dsfont} 
\usepackage{algorithm,algorithmicx,listings}        
\usepackage[noend]{algpseudocode}			              

\usepackage[usenames,dvipsnames]{color}       
\usepackage[font=small]{caption}
\usepackage{graphicx}
\usepackage[normalem]{ulem}	                  
\usepackage{extarrows}

\usepackage[breaklinks=true, letterpaper=true, colorlinks,%
            bookmarks=true, citecolor=Black, urlcolor=Violet]{hyperref}

\newcommand{\longeq}[2]{\xlongequal[\!#2\!]{\!#1\!}}
\newcommand{\longineq}[3]{\overset{#1}{\underset{#2}{#3}}}

\def\argmax{\mathop{\arg\max}\limits}	%

\DeclareMathOperator{\tr}{tr}

\newcommand{\myparagraph}[1]{\smallskip\noindent\textbf{\emph{#1.}}}

\renewcommand{\IEEEproofindentspace}{1\parindent}

\newtheorem{theorem}{Theorem}

\newtheorem*{assumption*}{Assumption}

\newtheorem*{problem*}{Problem}

\newtheorem{lemma}{Lemma}

\begin{document}
\title{\LARGE \bf Joint Estimation and Localization in Sensor Networks}

\author{Nikolay Atanasov, Roberto Tron, Victor M. Preciado, and George J. Pappas
\thanks{This work has been submitted to the IEEE for possible publication. Copyright may be transferred without notice, after which this version may no longer be accessible.}%
\thanks{This work was supported by ONR-HUNT grant N00014-08-1-0696 and by TerraSwarm, one of six centers of STARnet, a Semiconductor Research Corporation program sponsored by MARCO and DARPA.}%
\thanks{N. Atanasov, V. Preciado, and G. Pappas are with the Department of Electrical and Systems Engineering, University of Pennsylvania, Philadelphia, PA 19104, {\tt\small\{atanasov, preciado, pappasg\}@seas.upenn.edu}.}%
\thanks{R. Tron is with the Department of Computer and Information Science, University of Pennsylvania, Philadelphia, PA 19104, {\tt\small tron@seas.upenn.edu}}%
}
\maketitle

\begin{abstract}
This paper addresses the problem of collaborative tracking of dynamic targets in wireless sensor networks. A novel distributed linear estimator, which is a version of a \textit{distributed Kalman filter}, is derived. We prove that the filter is mean square consistent in the case of static target estimation. When large sensor networks are deployed, it is common that the sensors do not have good knowledge of their locations, which affects the target estimation procedure. Unlike most existing approaches for target tracking, we investigate the performance of our filter when the sensor poses need to be estimated by an auxiliary localization procedure. The sensors are localized via a distributed Jacobi algorithm from noisy relative measurements. We prove strong convergence guarantees for the localization method and in turn for the joint localization and target estimation approach. The performance of our algorithms is demonstrated in simulation on environmental monitoring and target tracking tasks.
\end{abstract}

\input{tex/Introduction.tex}

\input{tex/Problem.tex}

\input{tex/Estimation.tex}

\input{tex/Localization.tex}

\input{tex/Joint.tex}

\input{tex/Conclusion.tex}

\appendices
\input{tex/Appendix.tex}
\bibliographystyle{IEEEtran}
\bibliography{bib/ref.bib}

\end{document}

%% file: tex/Introduction.tex
\section{Introduction}
A central problem in networked sensing systems is the estimation and tracking of the states of dynamic phenomena of interest that evolve in the sensing field. Potential applications include environmental monitoring \cite{Mainwaring_WSN02,Leonard_IEEE07}, surveillance and reconnaissance \cite{Yan_SenSys03, Lu_TWC08}, social networks \cite{Molavi_GEB12}. In most situations, individual sensors receive partially informative measurements which are insufficient to estimate the target state in isolation. The sensors need to engage in information exchange with one another and solve a distributed estimation problem. To complicate matters, it is often the case that the sensors need to know their own locations with respect to a common reference in order to utilize the target measurements meaningfully. Hence, in general, the sensors face a joint \emph{localization} and \emph{estimation} problem. Virtually all existing work in distributed target estimation assumes implicitly that the localization problem is solved, while all the literature on localization does not consider the effect of the residual errors on a common estimation task. The goal of this paper is to show that the two problems can be solved jointly, and that, with simple measurement models, the resulting estimates have strong convergence guarantees.

\myparagraph{Assumptions and contributions} 
We assume that the sensors obtain linear Gaussian measurements of the target state and repeated sequential measurements of their relative positions along the edges of a graph. Our contributions are as follows:
\begin{itemize}
	\item We derive a distributed linear estimator for tracking dynamic targets. We prove that the filter is mean-square consistent in the case of a static target.
	\item We provide a distributed algorithm for sensor localization from sequential relative measurements and prove mean-square and strong consistency.
	\item We prove mean-square consistency of the joint localization and target estimation procedure.
\end{itemize}

\myparagraph{Related work}
Our target estimation algorithm was inspired by Rahnama Rad and Tahbaz-Salehi \cite{Alireza_CDC10}, who propose an algorithm for distributed \emph{static} parameter estimation using nonlinear sensing models. We specialize their model to heterogeneous sensors with linear Gaussian observations, show stronger convergence results (mean-square consistency instead of weak consistency), and then generalize the solution to \emph{dynamic} targets. Our filter is similar to the Kalman-Consensus \cite{Olfati-Saber_CDC07, Olfati-Saber_CDC09} and the filter proposed by Khan et al. \cite{Khan_CDC10, Shahrampour_NIPS13}. Khan et al. \cite{Khan_CDC10} show that a dynamic target can be tracked with bounded error if the norm of the target system matrix is less than the network tracking capacity. Shahrampour et al. \cite{Shahrampour_NIPS13} quantify the estimation performance using a global loss function and show that the asymptotic estimation error depends on its decomposition. Kar et al. \cite{Kar_TIT12} study distributed static parameter estimation with nonlinear observation models and noisy inter-sensor communication. Related work also includes \cite{Cortes_TAC09}, which combines the Jacobi over-relaxation method with dynamic consensus to compute distributed weighted least squares.

Our localization algorithm follows the lines of the Jacobi algorithm, first proposed for localization in sensor networks by Barooah and Hespanha \cite{Barooah_CS07,Barooah_PhD}. In contrast with their approach, we consider repeated relative measurements and show strong convergence guarantees for the resulting sequential localization algorithm. Other work in sensor network localization considers nonlinear and less informative measurement models than those used in this paper. For instance \cite{aspnes2006theory, Moore_SenSys04,so2007theory,Priyantha_ENSS03} address the problem of localization using range-only measurements, which is challenging because a graph with specified edge lengths can have several embeddings in the plane. Khan et al. \cite{Khan_TSP09} introduce a distributed localization (DILOC) algorithm, which uses the barycentric coordinates of a node with respect to its neighbors and show its convergence via a Markov chain. Diao et al. \cite{Diao_ASCC13} relax the assmuption of DILOC that all nodes must be inside the convex hull of the anchors. Localization has also been considered in the context of camera networks \cite{Tron:CDC09}.


\myparagraph{Paper organization} The joint localization and estimation problem is formulated precisely in Sec. \ref{sec:problem}. The distributed linear estimator for target tracking is derived in Sec. \ref{sec:estimation} assuming known sensor locations. A distributed Jacobi algorithm is introduced in Sec. \ref{sec:localization} to localize the sensors using relative measurements when the true locations are unknown. Mean-square and strong consistency are proven. In Sec. \ref{sec:joint_loc_estm}, we show that the error of the target estimator, when combined with the localization procedure, remains arbitrarily small. All proofs are provided in the Appendix.

%% file: tex/Problem.tex
\section{Problem Formulation}
\label{sec:problem}
\begin{figure*}[htb!]
\centering
\begin{minipage}[t]{.36\textwidth}
  \centering
  \includegraphics[width=\linewidth]{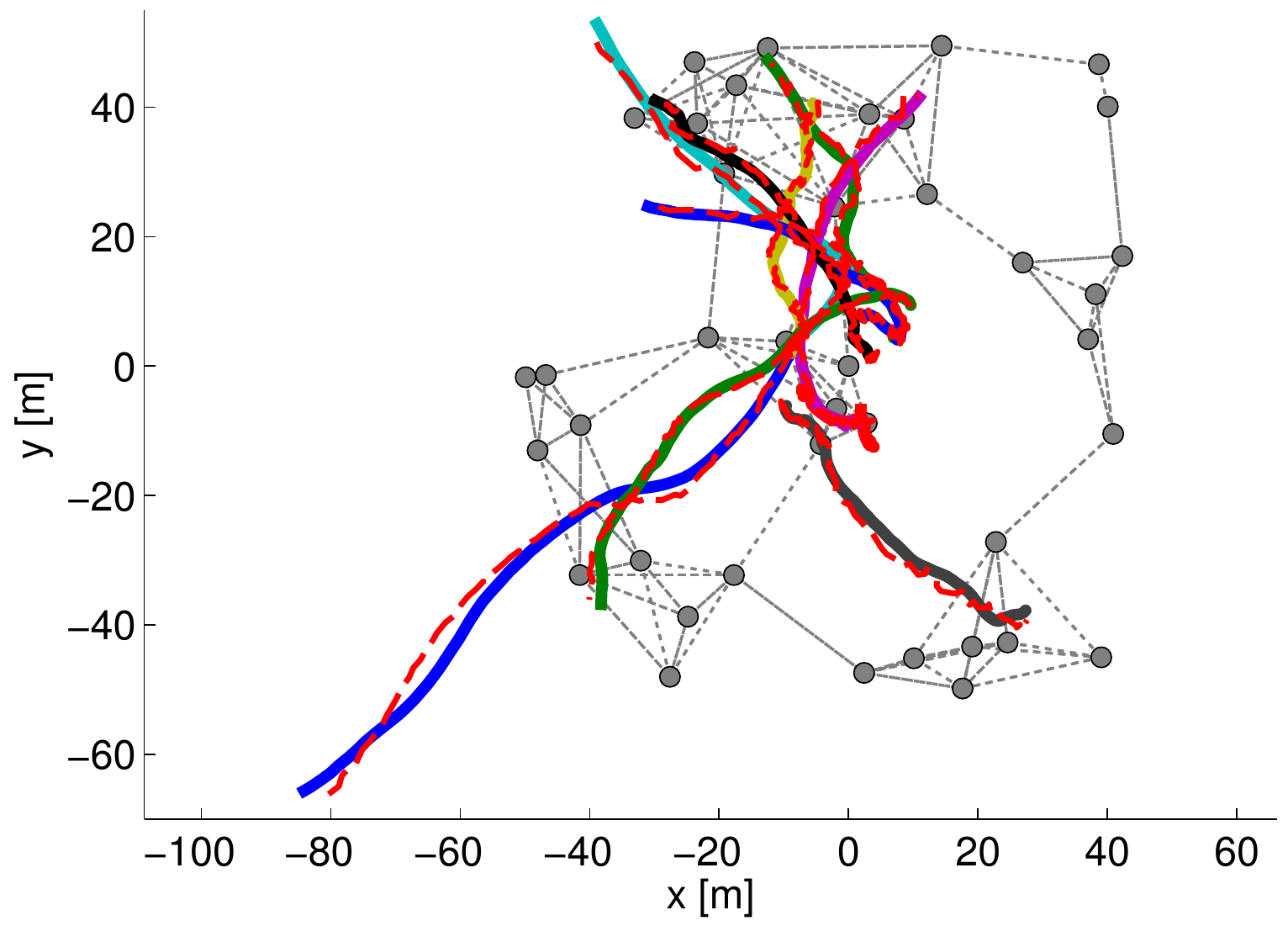}
  \captionof{figure}{A realization of the target tracking scenario in which a sensor network with 40 nodes tracks 10 mobile targets via range and bearing measurements}
  \label{fig:track_inst}
\end{minipage}%
\begin{minipage}{.02\textwidth}
$ $
\end{minipage}%
\begin{minipage}[t]{.60\textwidth}
  \centering
  \includegraphics[width=\linewidth]{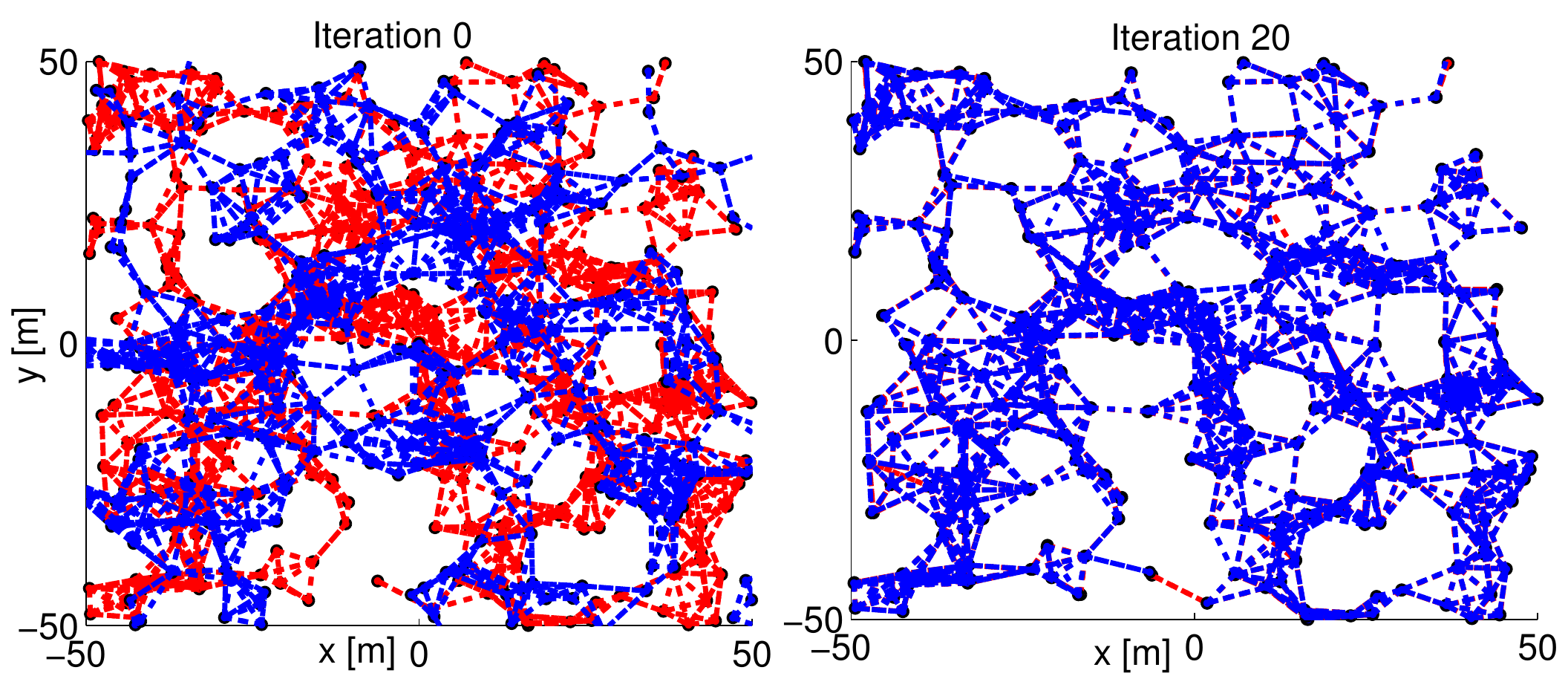}
  \captionof{figure}{Initial and final (after 20 steps) node locations estimated by the distributed localization algorithm on a randomly generated graph with 300 nodes ans 1288 edges}
  \label{fig:loc_inst}
\end{minipage}
\end{figure*}

Consider a static sensor network composed of $n$ sensors with configurations $\{x_1,\ldots,x_n\} \subset \mathcal{X} \cong \mathbb{R}^d$. The configuration of a sensor may include its position, orientation, and other operational parameters but we will refer to it, informally, as the sensor's location. The communication network interconnecting the sensors is represented by an undirected graph $G = (V,E)$ with vertices $V :=\{1,\ldots,n\}$ corresponding to the sensors and $|E|= m$ edges. An edge $(j,i) \in E$ from sensor $j$ to sensor $i$ exists if they can communicate. The set of nodes (neighbors) connected to sensor $i$ is denoted by $\mathcal{N}_i$.

The task of the sensors is to estimate and track the state $y(t) \in \mathcal{Y} \cong \mathbb{R}^{d_y}$ of a phenomenon of interest (target), where $\mathcal{Y}$ is a convex set. The target evolves according to the following \textit{target motion model}:
\begin{equation}
\label{eq:tmm}
y(t+1) = Fy(t) + \eta(t), \qquad \eta(t) \sim \mathcal{N}(0,W),
\end{equation}
where $\eta(t)$ is the process noise, whose values at any pair of times are independent. Sensor $i$, depending on its location $x_i$, can obtain a measurement $z_{i}(t)$ of the target state $y(t)$ at time $t$ according to the following \textit{sensor observation model}:
\begin{equation}
\label{eq:som}
z_{i}(t) = H_i(x_i)y(t) + v_{i}(t,x_i), \;\, v_{i}(t,x_i) \sim \mathcal{N}(0,V_i(x_i)),
\end{equation}
where $v_{i}(t,x_i)$ is a sensor-state-dependent measurement noise specific to sensor $i$, which is independent at any pair of times and across different sensors. The measurement noise is independent of the target noise $\eta(t)$ as well. The signals, $z_i(t)$, observed by a single sensor, although potentially informative, do not reveal the target state completely, i.e. each sensor faces a local identification problem. We assume, however, that the target is observable if one has access to the signals received by all sensors.

The sensors need to know their locations in order to use the signals $z_i(t)$ to estimate the targets state. However, when large sensor networks are deployed, it is common that the sensors do not have good knowledge of their positions but instead only a rough estimate (prior). We suppose that each sensor has access to noisy relative measurements of the positions of its neighbors\footnote{The graphs describing the communication and the relative measurement topologies might be different in practice. However, we assume that they are the same in order to simplify the presentation.}, which can be used to localize the sensors. In particular, at time $t$ sensor $i$ receives the following noisy relative configuration measurement from its neighbor $j$:
\begin{equation}
\label{eq:rel_meas}
s_{ij}(t) = x_j - x_i + \epsilon_{ij}(t), \qquad \epsilon_{ij}(t) \sim \mathcal{N}(0,\mathcal{E}_{ij}),
\end{equation}
where $\epsilon_{ij}(t)$ is a measurement noise, which is independent at any pair of times and across sensor pairs. The relative measurement noises are independent of the target measurement and motion noises too. Since there is translation ambiguity in the measurements \eqref{eq:rel_meas} we assume that all sensors agree to localize themselves in the reference frame of sensor 1. The location estimates can then be used in place of the unknown sensor positions during the target estimation procedure. The joint localization and estimation problem is summarized below.

\begin{problem*}[Joint Estimation and Localization]
The task of each sensor $i$ is to construct estimators $\hat{x}_{i}(t)$ and $\hat{y}_{i}(t)$ of its own location $x_i$ and of the target state $y$ in a distributed manner, i.e. using information only from its neighbors and the measurements $\{s_{ij}(t) \mid j \in \mathcal{N}_i\}$ and $\{z_{i}(t)\}$.
\end{problem*}

To illustrate the results, we use two scenarios which fit our models throughout the paper. The first is environmental monitoring problem in which a sensor network of remote methane leak detectors (RMLD), based on tunable diode laser absorption spectroscopy, is deployed to estimate the methane concentration in a landfill. The methane field is assumed static (i.e. $F = I_{d_y}, W = 0$) and can be modeled by discretizing the environment into cells and representing the gas concentration with a discrete Gaussian random field, $y \in \mathbb{R}^{d_y}$ (See Fig.\ref{fig:field}). It was verified experimentally in \cite{Hernandez_ICRA13} that the RMLD sensors fit the linear model in \eqref{eq:som}. Second, we consider tracking a swarm of mobile vehicles via a sensor network using range and bearing measurements (See Fig. \ref{fig:track_inst}). The position $(y_j^1, y_j^2) \in \mathbb{R}^2$ and velocity $(\dot{y}_j^1, \dot{y}_j^2) \in \mathbb{R}^2$ of the $j$th target have discretized double integrator dynamics driven by Gaussian noise:
\begin{gather*}
y_j(t+1) \!=\! \begin{bmatrix}
  I_2 & \tau I_2\\
  0 & I_2
\end{bmatrix} y_j(t) + \eta_j(t), \;\; W := q \textstyle{\begin{bmatrix}
  \frac{\tau^3}{3} I_2 & \frac{\tau^2}{2} I_2\\
  \frac{\tau^2}{2} I_2 & \tau I_2
\end{bmatrix}},
\end{gather*}
where $y_j = [y_j^1, y_j^2, \dot{y}_j^1, \dot{y}_j^2]^T$ is the $j$-th target state, $\tau$ is the sampling period is $sec$, and $q$ is a diffusion strength measured in $(\frac{m}{sec^2})^2 \frac{1}{Hz}$. Each sensor in the network takes noisy range and bearing measurements of the target's position:
\begin{equation}
z_{ij}(t) = \begin{bmatrix}
\sqrt{(y_j^1 - x_i^1)^2 + (y_j^2-x_i^2)^2}\\
\arctan\bigl((y_j^2-x_i^2)/(y_j^1 - x_i^1)\bigr)
\end{bmatrix} + v(t,x_i,y_j),\label{eq:rb_model}
\end{equation}
where $x_i := (x_i^1,x_i^2) \in \mathbb{R}^2$ is the sensor's location and the noise $v$ grows linearly with the distance between the sensor and the target. The observation model is nonlinear in this case so we resort to linearization in order to apply our framework.




%% file: tex/Estimation.tex
\section{Distributed Target Estimation}
\label{sec:estimation}
We begin with the task of estimating and tracking the dynamic state of a target via the sensor network. For now we assume that the sensors know their positions and concentrate on the estimation task. We specializing the general parameter estimation scheme of Rahnama Rad and Tahbaz-Salehi \cite{Alireza_CDC10} to linear Gaussian observation models such as (\ref{eq:som}). We show that the resulting distributed linear filter is mean-square consistent\footnote{\label{ftn:consistency}A distributed estimator of a parameter $y$ is \textit{weakly consistent} if all estimates, $\hat{y}_{i}(t)$, converge in probability to $y$, i.e. $\displaystyle{\lim_{t\to\infty} \mathbb{P}\bigl(\|\hat{y}_{i}(t) - y\| \geq \epsilon \bigr) = 0}$ for any $\epsilon > 0$ and all $i$. It is \textit{mean-square consistent} if all estimates converge in $L^2$ to $y$, i.e. $\displaystyle{\lim_{t \to \infty} \mathbb{E} \bigl[ \|\hat{y}_{i}(t) - y\|^2 \bigr] = 0},\forall i$.} when the target is stationary. This result is stronger than the weak consistency\textsuperscript{\ref{ftn:consistency}} shown in the general non-Gaussian case in \cite[Thm.1]{Alireza_CDC10}.
Suppose for now that the target is stationary, i.e. $y := y(0) = y(1) = \ldots$. To introduce the estimation scheme from \cite{Alireza_CDC10}, suppose also that instead of the linear Gaussian measurements in (\ref{eq:som}), the sensor measurements $z_i(t)$ are drawn from a general distribution with conditional probability density function (pdf) $l_i(\cdot \mid y)$. As before, the signals observed by sensor $i$ are iid over time and independent from the observations of all other sensors. In order to aggregate the information provided to it over time - either through observations or communication with neighbors - each sensor $i$ holds and updates a pdf $p_{i,t}$ over the target state space $\mathcal{Y}$. Consider the following distributed estimation algorithm:
\begin{equation}
\label{eq:dist_estm}
\begin{aligned}
p_{i,t+1}(y) &= \xi_{i,t} l_i(z_i(t+1) \mid y) \prod_{j \in \mathcal{N}_i \cup \{i\}} \bigl(p_{j,t}(y) \bigr)^{\kappa_{ij}},\\
\hat{y}_{i}(t) &\in \argmax_{y \in \mathcal{Y}} p_{i,t}(y), 
\end{aligned}
\end{equation}
where $\xi_{i,t}$ is a normalization constant ensuring that $p_{i,t+1}$ is a proper pdf and $\kappa_{ij} > 0$ are weights such that $\sum_{j \in \mathcal{N}_i \cup \{i\}} \kappa_{ij} = 1$. The update is the same as the standard Bayes rule with the exception that sensor $i$ does not just use its own prior but a \textit{geometric average} of its neighbors' priors. Given a connected graph, the authors of \cite{Alireza_CDC10} show that (\ref{eq:dist_estm}) is weakly consistent under broad assumptions on the observation models $l_i$.

Next, we specialize the estimator in (\ref{eq:dist_estm}) to the linear Gaussian measurement model in (\ref{eq:som}). Let $\mathcal{G}(\omega,\Omega)$ denote a Gaussian distribution (in information space) with mean $\Omega^{-1}\omega$ and covariance matrix $\Omega^{-1}$. The quantities $\omega$ and $\Omega$ are conventionally called \textit{information vector} and \textit{information matrix}, respectively. Suppose that the pdfs $p_{i,t}$ of all sensors $i \in V$ at time $t$ are that of Gaussian distributions $\mathcal{G}(\omega_{i,t},\Omega_{i,t})$. We claim that the posteriors resulting from applying the update in (\ref{eq:dist_estm}) remain Gaussian.

\begin{lemma}[\text{\cite[Thm.2]{Barker_CMA95}}]
\label{lem:geom_mean_gauss}
Let $Y_i \sim \mathcal{G}(\omega_i,\Omega_i)$ for $i = 1,\ldots,n$ be a collection of random Gaussian vectors with associated weights $\kappa_i$. The weighted geometric mean, $\prod_{i=1}^n p_i^{\kappa_i}$, of their pdfs $p_i$ is proportional to the pdf of a random vector with distribution $\mathcal{G}\biggl(\sum_{i=1}^n \kappa_i\omega_i, \sum_{i=1}^n \kappa_i\Omega_i \biggr)$.
\end{lemma}
\begin{lemma}[\text{\cite[Thm.2]{Barker_CMA95}}]
\label{lem:bayes_rule_gauss}
Let $Y \sim \mathcal{G}(\omega,\Omega)$ and $\mathcal{V} \sim \mathcal{G}(0,V^{-1})$ be random vectors. Consider the linear transformation $Z = HY + \mathcal{V}$. The conditional distribution of $Y \mid Z = z$ is proportional to $\mathcal{G}(\omega + H^TV^{-1}z, \Omega + H^TV^{-1}H)$.
\end{lemma}

Lemma \ref{lem:geom_mean_gauss} says that if the sensor priors are Gaussian $\mathcal{G}(\omega_{i,t},\Omega_{i,t})$, then after applying the geometric averaging in (\ref{eq:dist_estm}) the resulting distribution will still be Gaussian and its information vector and information matrix will be weighted averages of the prior ones. Lemma \ref{lem:bayes_rule_gauss} says that after applying Bayes rule the distribution will remain Gaussian. Combining the two allows us to derive the following linear Gaussian version of the estimator in (\ref{eq:dist_estm}):
\begin{equation}
\label{eq:gauss_dist_estm}
\begin{aligned}
\omega_{i,t+1} &= \sum_{j \in \mathcal{N}_i \cup \{i\}} \kappa_{ij} \omega_{j,t} + H_i^T V_i^{-1} z_{i}(t),\\
\Omega_{i,t+1} &= \sum_{j \in \mathcal{N}_i \cup \{i\}} \kappa_{ij} \Omega_{j,t} + H_i^T V_i^{-1} H_i,
\end{aligned}
\end{equation}
where $H_i := H_i(x_i)$ and $V_i := V_i(x_i)$. The estimate of sensor $i$ at time $t$ of the true target state $y$ is:
\begin{equation}
\label{eq:param_estm}
\hat{y}_i(t) := \Omega_{i,t}^{-1}\omega_{i,t}.
\end{equation}
In this linear Gaussian case, we prove a strong result about the quality of the estimates.

\begin{theorem}
\label{thm:dist_estm_L2}
Suppose that the communication graph $G$ is connected and the matrix $\begin{bmatrix} H_1^T & \hdots & H_n^T \end{bmatrix}^T$ has rank $d_y$. Then, the estimates (\ref{eq:param_estm}) of all sensors converge in mean square to $y$, i.e. $\displaystyle{\lim_{t \to \infty} \mathbb{E} \bigl[ \|\hat{y}_{i}(t) - y\|_2^2\bigr] = 0}$ for all $i$.
\end{theorem}


The estimation procedure in (\ref{eq:gauss_dist_estm}), (\ref{eq:param_estm}) can be extended to track a dynamic target as in (\ref{eq:tmm}) by adding a local prediction step, same as that of the Kalman filter, at each sensor. The distributed linear filter is summarized in Alg. \ref{alg:dist_lin_filt} and Thm. \ref{thm:dist_estm_L2} guarantees its mean-square consistency for stationary targets. Its performance on dynamic targets was studied in the target tracking scenario introduced in Sec. \ref{sec:problem} and the results are presented in Fig. \ref{fig:track_inst} and Fig. \ref{fig:track_rmse}.

\begin{algorithm}[htb!]
\caption{Distributed Linear Estimator}
\begin{algorithmic}[0]
\footnotesize
\State \textbf{Input}: Prior $(\omega_{i,t}, \Omega_{i,t})$, messages $(\omega_{j,t}, \Omega_{j,t}),\forall j \in \mathcal{N}_i$, and measurement $z_i(t)$
\State \textbf{Output}: $(\omega_{i,t+1}, \Omega_{i,t+1})$
\begin{flalign*}
&\text{Update Step:}  &\omega_{i,t+1}&= \sum_{j \in \mathcal{N}_i \cup \{i\}} \kappa_{ij} \omega_{j,t} + H_i^T V_i^{-1} z_{i}(t)&\\
& &\Omega_{i,t+1}&= \sum_{j \in \mathcal{N}_i \cup \{i\}} \kappa_{ij} \Omega_{j,t} + H_i^T V_i^{-1} H_i&\\
&&\hat{y}_i(t+1) &= \Omega_{i,t+1}^{-1}\omega_{i,t+1}&\\
&\text{Prediction Step:} &\Omega_{i,t+1}&= (F \Omega_{i,t+1}^{-1} F^T + W)^{-1}&\\
&&\omega_{i,t+1}&= \Omega_{i,t+1}F\hat{y}_i(t+1)&
\end{flalign*}
\end{algorithmic}
\label{alg:dist_lin_filt}
\end{algorithm}


\begin{figure*}[htb!]
\centering
\begin{minipage}[t]{.62\textwidth}
  \centering
  \includegraphics[width=\linewidth]{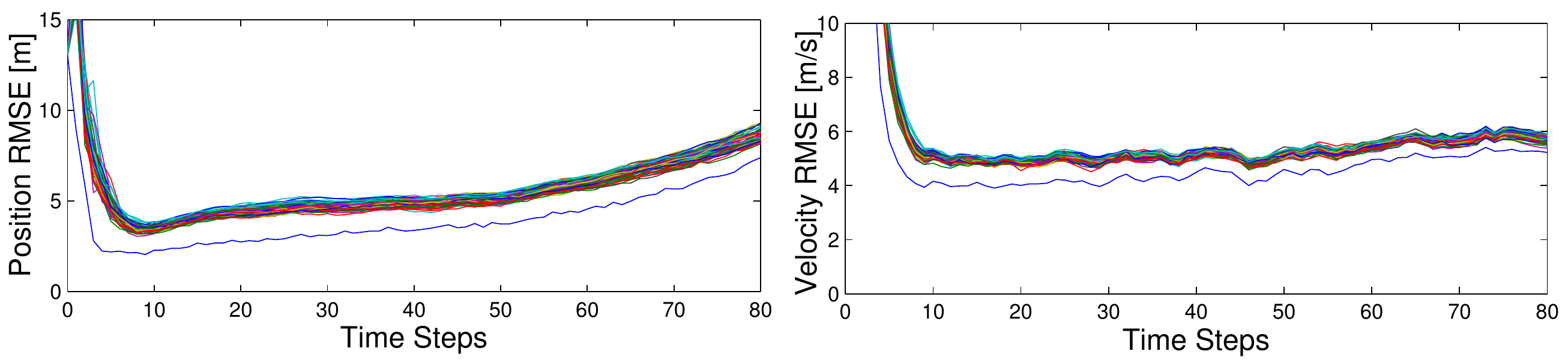}
  \captionof{figure}{Root mean squared error (RMSE) of the estimated target position and velocity obtained from averaging 50 simulated runs of the distributed linear estimator in the target tracking scenario (Fig. \ref{fig:track_inst}). The error increases because as targets move away from the sensor network, the covariance of the measurement noise grows linearly with distance. The errors of node 1 (blue) are lower because it was always placed at the origin and thus close to the starting target positions.} 
  \label{fig:track_rmse}
\end{minipage}%
\begin{minipage}[t]{.04\textwidth}
$ $
\end{minipage}%
\begin{minipage}[t]{.35\textwidth}
  \centering
  \includegraphics[width=\linewidth]{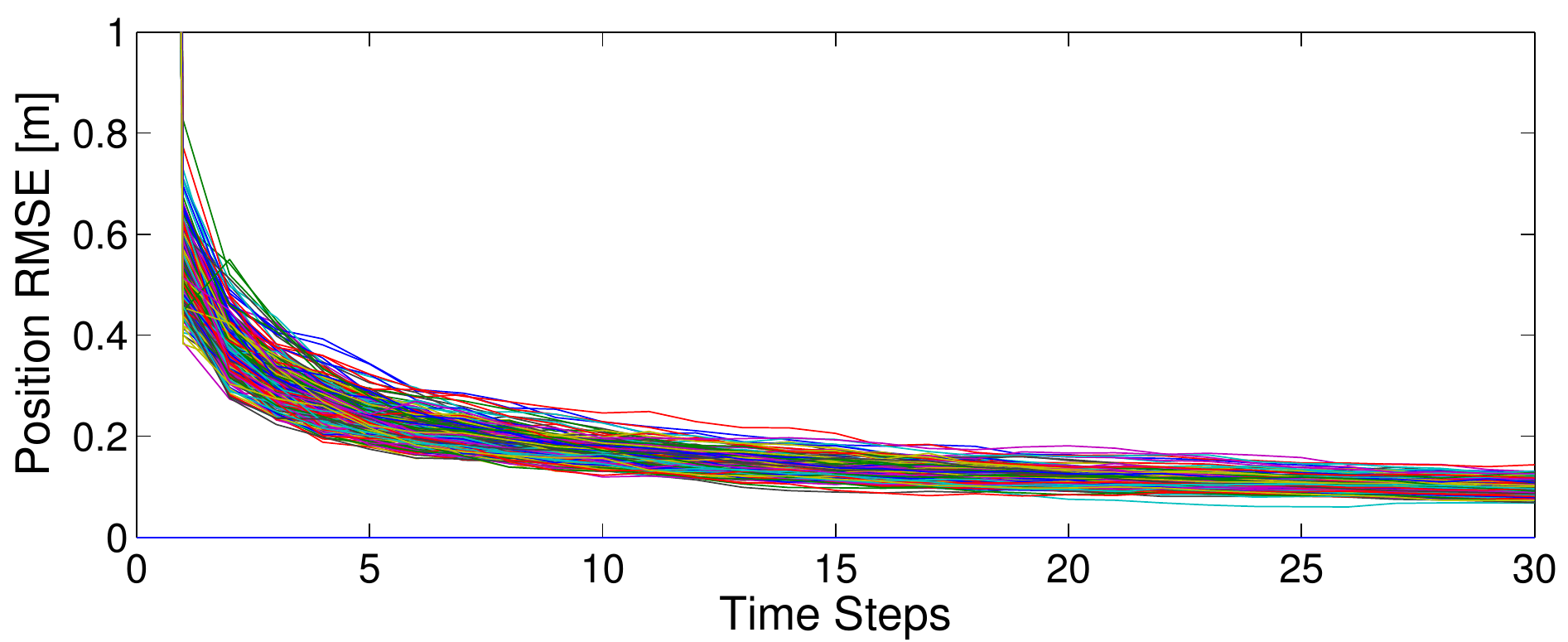}
  \captionof{figure}{Root mean squared error (RMSE) of the location estimates obtained from averaging 50 simulated runs of the distributed localization alogorithm  with randomly generated graphs with 300 nodes (e.g. Fig. \ref{fig:loc_inst})}
  \label{fig:loc_rmse}
\end{minipage}
\end{figure*}

%% file: tex/Localization.tex
\section{Localization from Relative Measurements}
\label{sec:localization}
Target tracking via the distributed estimator in Alg. \ref{alg:dist_lin_filt} requires knowledge of the true sensor locations. As mentioned earlier this is typically not the case, especially when large sensor networks are deployed. This section describes a method for localization from relative measurements (\ref{eq:rel_meas}), whose strong convergence guarantees can be used to analyze the convergence of a joint localization and estimation procedure. The relative measurements, received by all sensors at time $t$, can be written in matrix form as follows:
\[
s(t) = (B \otimes I_d)^T x + \epsilon(t),
\]
where $B \in \mathbb{R}^{n \times m}$ is the incidence matrix of the communication graph $G$. All sensors agree to localize relative to node 1 and know that $x_1 = 0$. Let $\tilde{B} \in \mathbb{R}^{(n-1) \times m}$ be the incidence matrix with the row corresponding to sensor 1 removed. Further, define $\mathcal{E} := \mathbb{E} [\epsilon(t)\epsilon(t)^T] = \mathbf{diag}(\mathcal{E}_1,\ldots,\mathcal{E}_m)$, where $\{\mathcal{E}_k\}$ is an enumeration of the noise covariances associated with the edges of $G$. Given $t$ measurements, the least squares estimate of $x$ leads to the classical Best Linear Unbiased Estimator (BLUE), given by:
\begin{equation}
\label{eq:BLUE}
 \hat{x}(t) := \bigl(\tilde{B} \mathcal{E}^{-1} \tilde{B}^T\bigr)^{-1} \tilde{B} \mathcal{E}^{-1}\sum_{\tau=0}^{t-1} s(\tau),
\end{equation}
where the inverse of $\tilde{B} \mathcal{E}^{-1} \tilde{B}^T$ exists as long as the graph $G$ is connected \cite{Barooah_CS07}. Among all linear estimators of $x$, BLUE has the smallest variance for the estimation error \cite{Mendel_EstTh95}. The computation in (\ref{eq:BLUE}) can be distributed via a Jacobi algorithm for solving a linear system as follows. Each sensor maintains an estimate $\hat{x}_i(t)$ of its own state at time $t$ and a history of the averaged measurements, $\sigma_i(t) := \frac{1}{t+1} \sum_{\tau=0}^t \sum_{j \in \mathcal{N}_i} \mathcal{E}_{ij}^{-1} s_{ij}(\tau)$, received up to time $t$. Given prior estimates $(\hat{x}_i(t),\sigma_i(t))$, the update of the distributed Jacobi algorithm at sensor $i$ is:
\begin{equation}
\label{eq:jacobi}
\begin{aligned}
  \hat{x}_i(t+1) &= \biggl( \sum_{j \in \mathcal{N}_i} \mathcal{E}_{ij}^{-1} \biggr)^{-1} \biggl(\sum_{j \in \mathcal{N}_i} \mathcal{E}_{ij}^{-1} \hat{x}_j(t) - \sigma_i(t)\biggr),\\
  \sigma_i(t+1) &= \frac{1}{t+1}\biggl(t \sigma_i(t) + \sum_{j \in \mathcal{N}_i} \mathcal{E}_{ij}^{-1} s_{ij}(t) \biggr).
\end{aligned}
\end{equation}

Barooah and Hespanha \cite{Barooah_CS07, Barooah_PhD} show that, with a \textit{single round} of relative measurements, the the Jacobi algorithm provides an unbiased estimate of $x$. Here, we incorporate repeated sequential measurements and prove much stronger performance guarantee.

\begin{theorem}
\label{thm:self-localization}
Suppose that the communication graph $G$ is connected. Then, the estimates $\hat{x}_{i}(t)$ of the sensor configurations in (\ref{eq:jacobi}) are mean-square and strongly consistent estimators of the true sensor states, i.e.:
\[
\lim_{t \to \infty} \mathbb{E}\bigl[ \|\hat{x}_{i}(t)-x_i\|_2^2\bigr] = 0, \; \mathbb{P}\bigl( \lim_{t \to \infty}\|\hat{x}_{i}(t)-x_i\|_2 = 0 \bigr) = 1,\forall i
\]
\end{theorem}

The performance of our distributed localization algorithm was analyzed on randomly generated graphs with 300 nodes. The location priors were chosen from a normal distribution with standard deviation of 5 meters from the true node positions. An instance of the localization task is illustrated in Fig. \ref{fig:loc_inst}, while the estimation error is shown in Fig. \ref{fig:loc_rmse}.


%% file: tex/Joint.tex

\section{Joint Localization and Estimation}
\label{sec:joint_loc_estm}
\begin{figure*}[ht!]
  \centering
  \includegraphics[width=0.24\linewidth]{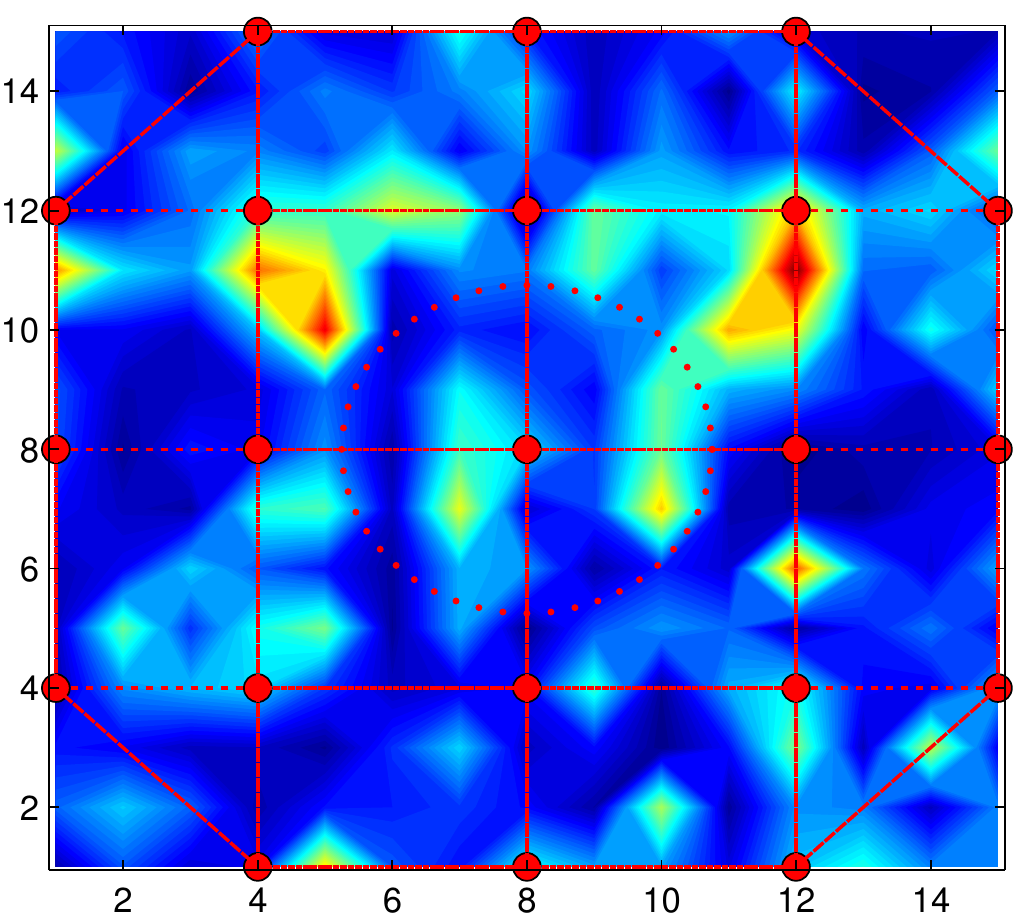}
  \raisebox{10pt}{\includegraphics[width=0.73\linewidth]{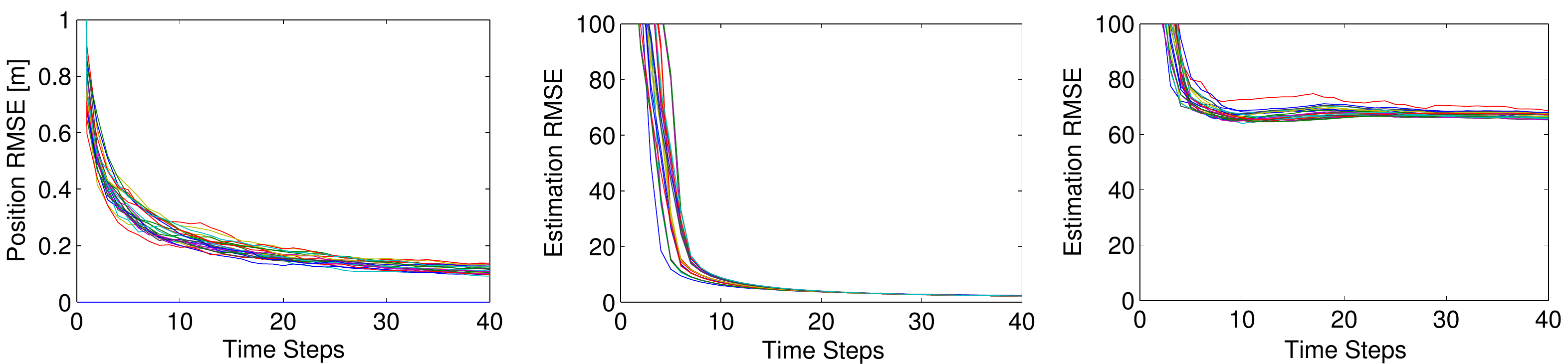}}
  \caption{Methane emission monitoring via a sensor network. The true (unknown) sensor locations (red dots), the sensing range (red circle), and a typical realization of the methane field are shown on the left. The root mean squared error (RMSE) of the location estimates and of the field estimates obtained from averaging 50 simulated runs of the joint localization and estimation algorithm with continuous sensor observation models are shown in the two middle plots. In an additional experiment, the sensors were placed on the boundaries of the cells of the discretized field. As the observation model for each sensor was defined in terms of the proximal environment cells, this made the model discontinuous. The rightmost plot illustrates that the field estimation error does not vanish when discontinuities are present.}
  \label{fig:field}
\end{figure*}

Having derived separate estimators for the sensor locations and the target state, we are ready to return to the original problem of joint localization and estimation. At time $t$, the location estimates $\{\hat{x}_{i}(t)\}$ in (\ref{eq:jacobi}) can be used in the target estimator (\ref{eq:gauss_dist_estm}), (\ref{eq:param_estm}) instead of the true sensor positions. It is important to analyze the evolution of the coupled estimation procedure because it is not clear that the convergence result in Thm. \ref{thm:dist_estm_L2} will continue to hold. Define the \textit{sensor information matrix} $M_i(x) := H_i(x)^T V_i(x)^{-1} H_i(x)$. In an analogy with the centralized Kalman filter, the sensor information matrix captures the amount of information added to the inverse of the covariance matrix during an update step of the Riccati map. From this point of view, it is natural to describe sensor properties in terms of the sensor information matrix. A regularity assumption which stipulates that nearby sensing locations provide similar information gain is necessary.

\begin{assumption*}[Observation Model Continuity]
The sensor information matrices $M_i(x)$ are bounded\footnote{There exists a constant $q$ such that $\|M_i(x)\| \leq q < \infty$ for all $i$ and $x$.} continuous functions of $x$ for all $i$.
\end{assumption*}

The following theorem ensures that the target state estimator retains its convergence properties when used jointly with the distributed localization procedure.

\begin{theorem}
\label{thm:joint_estm_L2}
Let $\{\hat{x}_i(t)\}$ be strongly consistent estimators of the sensor configurations, i.e. $\hat{x}_i(t) \xrightarrow{a.s.} x_i, \forall i$. Suppose that the communication graph $G$ is connected and the matrix $\begin{bmatrix} H_1(x_1)^T & \hdots & H_n(x_n)^T \end{bmatrix}^T$ has rank $d_y$. Let $\delta >0$ be arbitrary. If each sensor $i$ updates its target estimate $(\omega_{i,t},\Omega_{i,t})$ as follows:
\begin{equation}
\label{eq:joint_estm_L2}
\begin{aligned}
\omega_{i,t+1} &= \sum_{j \in \mathcal{N}_i \cup \{i\}} \kappa_{ij} \omega_{j,t} + \widehat{H}_{i,t}^T \widehat{V}_{i,t}^{-1} z_{i}(t),\\
\Omega_{i,t+1} &= \sum_{j \in \mathcal{N}_i \cup \{i\}} \kappa_{ij} \Omega_{j,t} + \widehat{H}_{i,t}^T \widehat{V}_{it}^{-1} \widehat{H}_{i,t},\\
\hat{y}_i(t+1) &= \bigl(\Omega_{i,t+1}+(t+1)\delta I_d\bigr)^{-1}\omega_{i,t+1},
\end{aligned}
\end{equation}
where $\widehat{H}_{i,t} := H_i(\hat{x}_{i}(t))$ and $\widehat{V}_{i,t} := V_i(\hat{x}_i(t))$, then the asymptotic mean-square error of target estimates is $O(\delta^2)$:
\begin{align*}
\lim_{t \to \infty} \mathbb{E} \bigl[ \|\hat{y}_{i}(t)-y\|_2^2\bigr] &= \delta^2 y^T \bigl(\sum_{j=1}^n\pi_jM_j(x_j) +\delta I\bigr)^{-2}y,
\end{align*}
for all $i$, where $y$ is the true target state and $x_j$ is the true position of sensor $j$.
\end{theorem}

The combined procedure specified by (\ref{eq:jacobi}) and (\ref{eq:joint_estm_L2}) provides a mean-square consistent way to estimate the sensor locations and the target state jointly. The performance of the joint algorithm was evaluated on the methane concentration estimation problem and the results are summarized in Fig. \ref{fig:field}.

%% file: tex/Conclusion.tex
\section{Conclusion}
This paper studied the problem of joint target tracking and node localization in sensor networks. A distributed linear estimator for tracking dynamic targets was derived. It was proven that the filter is mean-square consistent when estimating static states. Next, a distributed Jacobi algorithm was proposed for localization and its mean-square and almost sure consistency were shown. Finally, the combined localization and target estimation procedure was shown to have arbirarily small asymptotic estimation error. 

Future work will focus on strengthening the result in Thm. \ref{thm:joint_estm_L2} to mean-square consistency and relaxing the assumption of a strongly consistent localization procedure. Studying the relationship between our distributed linear estimator, the Kalman-Consensus filter \cite{Olfati-Saber_CDC09}, and the filter proposed by Khan et al. \cite{Khan_CDC10} is of interest as well.  




%% file: tex/Appendix.tex
\section*{Appendix A: Proof of Theorem \ref{thm:dist_estm_L2}}
\noindent Define the following:
\begin{align*}
\omega_t &:= \begin{bmatrix} \omega_{1t}^T & \hdots & \omega_{nt}^T \end{bmatrix}^T \quad &\Omega_t&:= \begin{bmatrix} \Omega_{1t}^T & \hdots & \Omega_{nt}^T \end{bmatrix}^T\\
M_{i} &:= H_{i}(x_i)^TV_i^{-1}(x_i)H_{i}(x_i)\quad &M&:=  \begin{bmatrix} M_{1}^T & \hdots & M_{n}^T \end{bmatrix}^T\\
\zeta(t) &:=\rlap{$\begin{bmatrix} H_{1}V_{1}^{-T}v_{1}(t)^T & \hdots & H_{n}V_{n}^{-T}v_n(t)^T \end{bmatrix}^T$.}
\end{align*}
The update equations of the filter (\ref{eq:gauss_dist_estm}) in matrix form are:
\begin{equation}
\label{eq:wO_system}
\begin{aligned}
\omega_{t+1} &= \bigl(\mathcal{K} \otimes I_{d_y}\bigr)\omega_t + My + \zeta(t),\\
\Omega_{t+1} &= \bigl(\mathcal{K} \otimes I_{d_y}\bigr)\Omega_t + M,
\end{aligned}
\end{equation}
where $\mathcal{K} =[\kappa_{ij}]$ with $\kappa_{ij} = 0$ if $j \notin \mathcal{N}_i \cup \{i\}$ is a stochastic matrix. The solutions of the linear systems are:
\begin{align*}
\omega_t &= \bigl(\mathcal{K} \otimes I_{d_y}\bigr)^t\omega_0 + \sum_{\tau=0}^{t-1}\bigl(\mathcal{K} \otimes I_{d_y}\bigr)^{t-1-\tau}\biggl(M y + \zeta(\tau)\biggr),\\
\Omega_t &= \bigl(\mathcal{K} \otimes I_{d_y}\bigr)^t \Omega_0 + \sum_{\tau=0}^{t-1}\bigl(\mathcal{K} \otimes I_{d_y}\bigr)^{t-1-\tau} M.
\end{align*}
Looking at the $i$-th components again, we have:
\begin{align*}
\frac{\omega_{it}}{t+1} &:=\frac{1}{t+1} \sum_{j=1}^n \bigl[\mathcal{K}^{t}\bigr]_{ij}\omega_{j0} +\\
&\quad\; \frac{1}{t+1} \sum_{\tau=0}^{t-1} \sum_{j=1}^n \bigl[\mathcal{K}^{t-\tau-1}\bigr]_{ij}(M_jy + H_j^T V_j^{-1}v_{j}(\tau)),\\
\frac{\Omega_{it}}{t+1} &:= \frac{1}{t+1} \sum_{j=1}^n \bigl[\mathcal{K}^{t}\bigr]_{ij}\Omega_{j0} +\frac{1}{t+1} \sum_{\tau=0}^{t-1} \sum_{j=1}^n \bigl[\mathcal{K}^{t-\tau-1}\bigr]_{ij}M_{j}.
\end{align*}
Define the following to simplify the notation:
\begin{equation}
\begin{aligned}
 g_{it} &:= \textstyle{\frac{1}{t+1} \sum_{j=1}^n \bigl[\mathcal{K}^{t}\bigr]_{ij}\omega_{j0}},\\
G_{it}&:= \textstyle{\frac{1}{t+1} \sum_{j=1}^n \bigl[\mathcal{K}^{t}\bigr]_{ij}\Omega_{j0}},\\
\phi_{it} &:=\textstyle{\frac{1}{t+1} \sum_{\tau=0}^{t-1} \sum_{j=1}^n \bigl[\mathcal{K}^{t-\tau-1}\bigr]_{ij} H_j^T V_j^{-1}v_{j}(\tau)},\\
C_{it} &:= \textstyle{\frac{1}{t+1} \sum_{\tau=0}^{t-1} \sum_{j=1}^n \bigl[\mathcal{K}^{t-\tau-1}\bigr]_{ij}M_{j}},\\
b_{it} &:= g_{it} - G_{it}y, \qquad B_{it} := \textstyle{\frac{1}{t+1}\Omega_{it}}.
\end{aligned}
\label{eq:shorthand-notation}
\end{equation}
With the shorthand notation:
\begin{equation}
\label{eq:wO_evolution}
  \frac{\omega_{it}}{t+1} = g_{it}+\phi_{it}+C_{it}y, \qquad B_{it}=\frac{\Omega_{it}}{t+1} =  G_{it}+C_{it},
\end{equation}
where $\phi_{it}$ is the only random quantity. Its mean is zero because the measurement noise has zero mean, while its covariance is:
\begin{align}
  &\mathbb{E} [\phi_{it} \phi_{it}^T] =\frac{1}{(t+1)^2}\mathbb{E}\biggl[\biggl(\sum_{\tau=0}^{t-1} \sum_{j=1}^n \bigl[\mathcal{K}^{t-\tau-1}\bigr]_{ij}H_{j}^T V_j^{-1} v_{j}(\tau)\biggr)\notag\\
  &\qquad \times \biggl(\sum_{s=0}^{t-1}\sum_{\eta=1}^n \bigl[\mathcal{K}^{t-s-1}\bigr]_{i\eta}H_{\eta}^TV_\eta^{-1} v_{\eta}(s)\biggr)^T\biggr]\notag\\
  &\!=\!\frac{1}{(t+1)^2}\!\sum_{j=1}^n \sum_{\tau=0}^{t-1} \bigl[\mathcal{K}^{t-\tau-1}\bigr]_{ij}^2H_j^T V_j^{-1} \mathbb{E}[v_{j}(\tau)v_{j}(\tau)^T]V_j^{-1}H_{j}\notag\\
  &= \frac{1}{(t+1)^2}\sum_{j=1}^n\sum_{\tau=0}^{t-1} \bigl[\mathcal{K}^{t-\tau-1}\bigr]_{ij}^2M_j \preceq \frac{1}{t+1} C_{it},\label{eq:covariance-phiit}
\end{align}
where the second equality uses the fact that $v_j(\tau)$ and $v_\eta(s)$ are independent unless the indices coincide, i.e. $\mathbb{E} v_j(\tau) v_\eta(s)^T = \delta_{\tau s} \delta_{j \eta} V_j$. The L{\"o}wner ordering inequality in the last step uses that $0\leq\bigl[\mathcal{K}^{t-\tau-1}\bigr]_{ij}\leq 1$ and $M_j \succeq 0$.

Since $G$ is connected, $\mathcal{K}$ corresponds to the transition matrix of an aperiodic irreducible Markov chain with a unique stationary distribution $\pi$ so that $\mathcal{K}^t \to \pi \mathbf{1}^T$ with $\pi_j >0$. This implies that, as $t \to \infty$, the numerators of $g_{it}$ and $G_{it}$ remain bounded and therefore $g_{it} \to 0$ and $G_{it} \to 0$. Since Ces\'aro means preserve convergent sequences and their limits:
\[
\frac{1}{t+1} \sum_{\tau=0}^{t-1} \bigl[\mathcal{K}^{t-\tau-1}\bigr]_{ij} \to \pi_j, \quad \forall i,
\]
which implies that $C_{it} \to \sum_{j=1}^n \pi_j M_j$. The full-rank assumption on $\begin{bmatrix} H_1^T & \hdots & H_n^T \end{bmatrix}^T$ and $\pi_j>0$ guarantee that $\sum_{j=1}^n \pi_j M_j$ is positive definite. Finally, consider the mean squared error:
\begin{align*}
&\mathbb{E}\bigl[(\hat{y}_i(t) - y)^T(\hat{y}_i(t) - y) \bigr]\\
&= \mathbb{E} \biggl\| \biggl(\frac{\Omega_{it}}{t+1}\biggr)^{-1}\frac{\omega_{it}}{t+1} -  \biggl(\frac{\Omega_{it}}{t+1}\biggr)^{-1}\biggl(\frac{\Omega_{it}}{t+1}\biggr)y \biggr\|_2^2\\
&= \mathbb{E} \bigl\| B_{it}^{-1} \bigl(g_{it} + C_{it}y + \phi_{it} - (G_{it} + C_{it})y \bigr)  \bigr\|_2^2\\
&= \mathbb{E} \|B_{it}^{-1}(b_{it} +\phi_{it})\|_2^2\\
&= \mathbb{E}\biggl[b_{it}^T B_{it}^{-T} B_{it}^{-1} b_{it} + 2b_{it}^TB_{it}^{-T} B_{it}^{-1}\phi_{it} + \phi_{it}^TB_{it}^{-T} B_{it}^{-1}\phi_{it}\biggr]\\
&\longeq{(a)}{} b_{it}^T B_{it}^{-T} B_{it}^{-1} b_{it} + \tr( B_{it}^{-1} \mathbb{E} [\phi_{it} \phi_{it}^T] B_{it}^T)\\
&\longineq{(b)}{}{\leq} b_{it}^T B_{it}^{-T}B_{it}^{-1}b_{it} +\frac{1}{t+1}\tr(B_{it}^{-1}C_{it} B_{it}^{-T}) \to 0,
\end{align*}
where $(a)$ holds because the first term is deterministic, while the cross term contains $\mathbb{E}[\phi_{it}] = 0$. Inequality $(b)$ follows from \eqref{eq:covariance-phiit}. In the final step, as shown before $B_{it}^{-1} \to \bigl(\sum_{j=1}^n \pi_j M_j \bigr)^{-1}$ and $C_{it} \to \sum_{j=1}^n \pi_j M_j \succ 0$ remain bounded, while $b_{it} \to 0$ and $1/(t+1) \to 0$.\hfill$\qed$

\section*{Appendix B: Proof of Theorem \ref{thm:self-localization}}
\label{app:A}
Define the \textit{generalized (matrix-weighted) degree matrix} $D \in \mathbb{R}^{nd \times nd}$ of the graph $G$ as a block-diagonal matrix with $D_{ii} := \sum_{j \in \mathcal{N}_i} \mathcal{E}^{-1}_{ij}$. Since $\mathcal{E}_{ij} \succ 0$ for all $\{i,j\}\in E$, the generalized degree matrix is positive definite, $D \succ 0$. Define also the \textit{generalized adjacency matrix} $A \in \mathbb{R}^{nd \times nd}$ as follows:
\[
A_{ij} := \begin{cases}
            \mathcal{E}^{-1}_{ij} & \text{if } \{i,j\} \in E,\\
            0 & \text{else}.
          \end{cases}
\]
The \textit{generalized Laplacian} and the \textit{generalized signless Laplacian} of $G$ are defined as $L := D-A$ and $|L| := D + A$, respectively. Further, let $R := (B\otimes I_d)^T \in \mathbb{R}^{md \times nd}$ and define the block-diagonal matrix $\mathcal{E} \in \mathbb{R}^{md \times md}$ with blocks $\mathcal{E}_{ij}$ for $\{i,j\} \in E$. It is straightforward to verify that $L = R^T \mathcal{E}^{-1} R \succeq 0$ and $|L| = (|B| \otimes I_d) \mathcal{E}^{-1} (|B| \otimes I_d)^T \succeq 0$, where $|B| \in \mathbb{R}^{n \times m}$ is the signless incidence matrix of $G$. Let $\tilde{B} \in \mathbb{R}^{(n-1) \times m}$ and $\tilde{R} \in \mathbb{R}^{md \times (n-1)d}$ be the matrices resulting after removing the row corresponding to sensor $1$ from $B$. Similarly, let $\tilde{D}, \tilde{A}, \tilde{L}, |\tilde{L}|\in \mathbb{R}^{(n-1)d \times (n-1)d}$ denote the generalized degree, adjacency, Laplacian, and signless Laplacian matrices with the row and column corresponding to sensor $1$ removed. Thm. 2.2.1 in \cite{Barooah_PhD} shows that $\tilde{L} \succ 0$ provided that $G$ is connected. The same approach can be used to show that $|\tilde{L}| \succ 0$. Let $\tilde{x} \in \mathbb{R}^{(n-1)d}$ be the locations of sensors $2,\ldots,n$ in the reference frame of sensor $1$ and $\hat{x}(t) \in \mathbb{R}^{(n-1)d}$ be their estimates at time $t$ obtained from (\ref{eq:jacobi}). The update in (\ref{eq:jacobi}) can be written in matrix form as follows:
\begin{equation}
\label{eq:jacobi_iter}
\tilde{D} \hat{x}(t+1) = \tilde{A} \hat{x}(t) + \tilde{R}^T \mathcal{E}^{-1} \biggl(\tilde{R} \tilde{x} + \frac{1}{t+1}\sum_{\tau = 0}^t \epsilon(\tau)\biggr).
\end{equation}
Define the estimation error at time $t$ as $e(t) \!:=\! \tilde{x} - \hat{x}(t)$ and let $u(t) \!:=\! \frac{1}{t+1}\sum_{\tau = 0}^t \epsilon(\tau)$. The dynamics of the error state can be obtained from (\ref{eq:jacobi_iter}):
\begin{align*}
e(t&+1) = \tilde{x} - \tilde{D}^{-1} \tilde{A} \hat{x}(t) - \tilde{D}^{-1} \tilde{L}\tilde{x} - \tilde{D}^{-1}\tilde{R}^T \mathcal{E}^{-1} u(t)\\
&= \tilde{x} - \tilde{D}^{-1} \tilde{A} \hat{x}(t) - \tilde{D}^{-1} \biggl(\tilde{D} - \tilde{A}\biggr) \tilde{x} - \tilde{D}^{-1}\tilde{R}^T \mathcal{E}^{-1} u(t)\\
&= \tilde{D}^{-1} \tilde{A} e(t) - \tilde{D}^{-1}\tilde{R}^T \mathcal{E}^{-1} u(t).
\end{align*}
The error dynamics are governed by a stochastic linear time-invariant system, whose internal stability depends on the eigenvalues of $\tilde{D}^{-1} \tilde{A}$. To show that the error dynamics are stable, we resort to the following lemma.
\begin{lemma}[\text{\cite[Lemma 4.2]{Ludwig_NM91}}]
\label{lem:spec_rad}
Let $L = D-A \in \mathbb{C}^{n \times n}$ be such that $D + D^* \succ 0$ and $L_\theta = D + D^* - (e^{i\theta} A + e^{-i\theta}A^*) \succ 0$ for all $\theta \in \mathbb{R}$. Then $\rho(D^{-1}A) < 1$.
\end{lemma}

Consider $\tilde{L}_\theta := 2(\tilde{D} - \cos(\theta) \tilde{A})$. If $\cos \theta = 0$, then $\tilde{L}_\theta = 2\tilde{D} \succ 0$. If $\cos \theta \in (0,1]$, then $\tilde{L}_\theta \succeq 2 \cos \theta \tilde{L} \succ 0$. Finally, if $\cos \theta \in [-1,0)$, then $\tilde{L}_\theta \succeq 2|\cos \theta| |\tilde{L}| \succ 0$. Thus, we can conclude that $\rho \bigl( \tilde{D}^{-1} \tilde{A} \bigr) < 1$. The proof of the theorem is completed by the following lemma with $\mathsf{F} := \tilde{D}^{-1} \tilde{A}$ and $\mathsf{G} :=  - \tilde{D}^{-1}\tilde{R}^T \mathcal{E}^{-1}$.

\begin{lemma}
Consider the discrete-time stochastic linear time-invariant system:
\begin{equation}
\label{eq:SLTI}
\textstyle{e(t+1) = \mathsf{F} e(t) + \mathsf{G} \frac{1}{t+1} \sum_{\tau=0}^t \epsilon(\tau)}
\end{equation}
driven by Gaussian noise $\epsilon(\tau) \sim \mathcal{N}(0, \mathcal{E})$, which is independent at any pair of times. If the spectral radius of $\mathsf{F}$ satisfies $\rho(\mathsf{F}) < 1$, then $e(t) \xrightarrow{a.s., L^2} 0$.
\end{lemma}

\begin{proof}
By the strong law of large numbers \cite[Thm.2.4.1]{Durrett_Prob10}, $u(t) := \frac{1}{t+1} \sum_{\tau=0}^t \epsilon(\tau)$ converges to $0$ almost surely. Let $\Omega$ be the set with measure $1$ on which $u(t)$ converges so that for any $\gamma > 0$, $\exists\; T \in \mathbb{N}$ such that $\forall t \geq T$, $\|u(t)\| \leq \gamma$. For realizations in $\Omega$, the solution to (\ref{eq:SLTI}) with initial time $T$ is:
\[
\textstyle{e(t) = \mathsf{F}^{t-T} e(T) + \sum_{\tau=T}^{t-1} \mathsf{F}^{t-\tau-1} \mathsf{G} u(\tau)}.
\]
Then, $\displaystyle{\|e(t)\| \leq \| \mathsf{F}^{t-T} e(T)\| + \sum_{\tau = T}^{t-1}  \bigl\|\mathsf{F}^{t - \tau -1}\bigr\| \|\mathsf{G}\| \gamma}$. Taking the limit of $t$ and using that $\mathsf{F}$ is stable, we have
\[
\lim_{t \to \infty} \|e(t)\| \leq \biggl(\sum_{\tau = 0}^{\infty} \bigl\|\mathsf{F}^{\tau}\bigr\|\biggr) \|\mathsf{G}\| \gamma.
\]
Since $\rho(\mathsf{F}) < 1$, the system is internally (uniformly) exponentially stable, which is equivalent to $\sum_{\tau=0}^\infty \|\mathsf{F}^\tau\| \leq \beta$ for some finite constant $\beta$ \cite[Ch.22]{Rugh_LinSys96}. Thus, $\lim_{t \to \infty} \|e(t)\| \leq \beta \|\mathsf{G}\|\gamma$, which can be made arbitrarily small by choice of $\gamma$. We conclude that $e(t) \to 0$ on $\Omega$ and consequently $e(t) \xrightarrow{a.s.} 0$.

Next, we show convergence in $L^2$. First, consider the propagation of the cross term $C(t) := (t+1)\mathbb{E}e(t) u(t)^T$. Note that $\mathbb{E} u(t) = 0$ and $\mathbb{E} u(t)u(t)^T = \frac{\mathcal{E}}{t+1}$. Using the fact that $\epsilon(t+1)$ is independent of $e(t)$ and $u(t)$ we have
\begin{align*}
C(t+1) &= \mathbb{E} \bigl(\mathsf{F}e(t) + \mathsf{G}u(t)\bigr) \bigl((t+1) u(t) + \epsilon(t+1) \bigr)^T\\
&= \mathsf{F} C(t) + (t+1) \mathsf{G} \mathbb{E} u(t)u(t)^T = \mathsf{F} C(t) + \mathsf{G}\mathcal{E}.
\end{align*}
The solution of the above linear time-invariant system is:
\[
\textstyle{C(t) = \mathsf{F}^t C(0) + \sum_{\tau=0}^{t-1} \mathsf{F}^{t-\tau-1}\mathsf{G}\mathcal{E}}
\]
and since $\mathsf{F}$ is stable:$\displaystyle{\lim_{t \to \infty}} \mathbb{E}e(t) u(t)^T \!\!\!=\!\!\displaystyle{\lim_{t \to \infty}} \frac{1}{t+1}\!\sum_{\tau=0}^{t-1} \!\mathsf{F}^\tau \mathsf{G}\mathcal{E} \!=\! 0$. Now, consider the second moment of the error:
\begin{align*}
&\Sigma(t+1) := \mathbb{E} e(t+1) e(t+1)^T =\\
&\textstyle{\mathsf{F} \Sigma(t) \mathsf{F}^T \!+\! \mathsf{F} \biggl(\mathbb{E} e(t) u(t)^T\biggr) \mathsf{G}^T \!+\! \mathsf{G} \biggl(\mathbb{E} u(t) e(t)^T\biggr) \mathsf{F}^T \!\!+\! \frac{1}{t+1} \mathsf{G}\mathcal{E}\mathsf{G}^T}\\
&= \mathsf{F} \Sigma(t) \mathsf{F}^T + Q(t),
\end{align*}
where $\displaystyle{Q(t) := \frac{1}{t+1} \biggl(\mathsf{F} C(t) \mathsf{G}^T + \mathsf{G}C(t)^T\mathsf{F}^T + \mathsf{G}\mathcal{E}\mathsf{G}^T\biggr)}$. As shown above $Q(t) \to 0$ as $t \to \infty$, i.e. for any $\delta >0$, $\exists\; T' \in \mathbb{N}$ such that $\forall t \geq T'$, $\|Q(t)\| \leq \delta$. With initial time $T'$,
\[
\Sigma(t) = \mathsf{F}^{t-T'} \Sigma(T') (\mathsf{F}^T)^{t-T'}  + \sum_{\tau = T'}^{t-1} \mathsf{F}^{t-\tau-1} Q(\tau) (\mathsf{F}^T)^{t - \tau - 1}
\]
for $t \geq T'$. Then:
\begin{align*}
\|\Sigma(t)\| &\leq \bigl\|\mathsf{F}^{t-T'}\bigr\|^2\|\Sigma(T')\| + \sum_{\tau = 0}^{t-T'-1} \|\mathsf{F}^{\tau}\|^2 \delta\\
&\leq \alpha^2 \mu^{2(t-T')} + \delta \alpha^2 \sum_{\tau = 0}^{t-T'-1} \mu^{2\tau},
\end{align*}
where the existence of the constants $\alpha > 0$ and $0 \leq \mu < 1$ is guaranteed by the stability of $\mathsf{F}$. We conclude that $\lim_{t \to \infty} \|\Sigma(t)\| \leq \frac{\delta \alpha^2}{1-\mu^2}$, which can be made arbitrarily small by choice of $\delta$. In other words, $e(t) \xrightarrow{L^2} 0$.
\end{proof}

\section*{Appendix C: Proof of Theorem \ref{thm:joint_estm_L2}}
\label{app:B}
We use the same notation and follow the same steps as in the proof of Thm. \ref{thm:dist_estm_L2}, except that now the terms $H_i, V_i, M_i, M, \zeta(t), \phi_{it}, C_{it}, B_{it}$ are time-varying and stochastic because they depend on the location estimates $\hat{x}_{i}(t)$. To emphasize this, we denote them by $\widehat{H}_{it}, \widehat{V}_{it}, \widehat{M}_{it}, \widehat{M}_t, \widehat{\zeta}(t), \widehat{\phi}_{it}, \widehat{C}_{it}, \widehat{B}_{it}$, where for example $\widehat{M}_{it} := M_i(\hat{x}_i(t))$. The same linear systems \eqref{eq:wO_system} describe the evolutions of $\omega_t$ and $\Omega_t$ except that they are stochastic now and \eqref{eq:wO_evolution} becomes:
\[
\frac{\omega_{it}}{t+1} = g_{it} + \widehat{C}_{it}y + \widehat{\phi}_{it}, \qquad \widehat{B}_{it} := \frac{\Omega_{it}}{t+1} = G_{it} + \widehat{C}_{it}.
\]
We still have that $\mathcal{K}^t \to \pi \mathbf{1}^T$ with $\pi_j > 0$. Also, $g_{it}$, $G_{it}$, and $b_{it}$ are still deterministic and converge to zero as $t \to \infty$. The following observations are necessary to conclude that $\widehat{C}_{it}$ still converges to $\sum_{j=1}^n \pi_j M_j$.
\begin{lemma}
\label{lem:H_L2_conv}
If $\hat{x}_i(t) \xrightarrow{a.s.} x_i$, then $\widehat{M}_{it} \xrightarrow{a.s., L^2} M_i$.
\end{lemma}
\begin{proof}
Almost sure convergence follows from the continuity of $M_i(\cdot)$ and the continuous mapping theorem \cite[Thm.3.2.4]{Durrett_Prob10}. $L^2$-convergence follows from the boundedness of $M_i(\cdot)$ and the dominated convergence theorem \cite[Thm.1.6.7]{Durrett_Prob10}.
\end{proof}
\begin{lemma}
\label{lem:cesaro_conv}
If $a_t \to a$ and $b_t \to b$, then $\frac{1}{t}\sum_{\tau=0}^{t-1} a_{t-\tau}b_\tau \to ab$. 
\end{lemma}
\begin{proof}
The convergence of $a_t$ implies its boundedness, $|a_t| \leq q < \infty$. Then, notice $ab=\frac{1}{t}\sum_{\tau=0}^{t-1} ab$ and
\begin{align*}
\biggl|\frac{1}{t}\sum_{\tau=0}^{t-1}a_{t-\tau}b_\tau &- ab\biggr| = \biggl|\frac{1}{t}\sum_{\tau=0}^{t-1}\bigl( a_{t-\tau}(b_\tau-b) +(a_{t-\tau} - a)b\bigr)\biggr|\\
&\leq \biggl|\frac{1}{t}\sum_{\tau=0}^{t-1} a_{t-\tau}(b_\tau-b)\biggr| + \biggl|\frac{1}{t}\sum_{\tau=0}^{t-1}(a_{t-\tau} - a)b\biggr|\\
&\leq \biggl|q\biggl(\frac{1}{t}\sum_{\tau=0}^{t-1}b_\tau-b\biggr)\biggr| + \biggl|\biggl(\frac{1}{t}\sum_{\tau=1}^{t}a_{\tau} - a\biggr)b\biggr|,
\end{align*}
where both terms converge to zero since Ces\'aro means preserve convergent sequences and their limits.
\end{proof}
Combining Lemma \ref{lem:H_L2_conv}, $\bigl[\mathcal{K}^t\bigr]_{ij} \to \pi_j$, and Lemma \ref{lem:cesaro_conv}, we have:
\[
\frac{1}{t+1} \sum_{\tau=0}^{t-1} \bigl[\mathcal{K}^{t-\tau-1}\bigr]_{ij}\widehat{M}_{j\tau} \xrightarrow{a.s.} \bigl[\pi \mathbf{1}^T\bigr]_{ij}M_j = \pi_jM_j.
\]
Moreover, $0 \leq [\mathcal{K}^t]_{ij} \leq 1$ and the boundedness of $\widehat{M}_{jt}$ imply, by the bounded convergence theorem \cite[Thm.1.6.7]{Durrett_Prob10}, that the sequence above converges in $L^2$ as well:
\begin{equation}
\label{eq:c_convergence}
\textstyle{\widehat{C}_{it} \xrightarrow{a.s., L^2} \sum_{j=1}^n \pi_j M_j \succ 0.}
\end{equation}
In turn, \eqref{eq:c_convergence} guarantees that:
\begin{equation}
\label{eq:B_convergence}
\textstyle{\widehat{B}_{it}^{-2} = \bigl(G_{it} + \widehat{C}_{it}\bigr)^{-2} \xrightarrow{a.s.} \bigl(\sum_{j=1}^n \pi_j M_j\bigr)^{-2}}
\end{equation}
but is not enough to ensure that $\mathbb{E} \bigl[\widehat{B}_{it}^{-2}\bigr]$ remains bounded as $t \to \infty$. The parameter $\delta>0$ is needed to guarantee the boundedness. In particular, define $\widehat{B}_{it}(\delta) := \widehat{B}_{it} + \delta I_{d_y}$. Then
\[
\widehat{B}_{it}(\delta)^{-2} = \bigl(G_{it} + \widehat{C}_{it} + \delta I_{d_y}\bigr)^{-2} \prec \frac{1}{\delta^2} I_{d_y}
\]
and by the bounded convergence theorem and (\ref{eq:B_convergence}):
\begin{equation}
\label{eq:invB_bounded}
\textstyle{\widehat{B}_{it}(\delta)^{-2} \xrightarrow{a.s., L^1} \bigl(\sum_{j=1}^n \pi_j M_j + \delta I_{d_y}\bigr)^{-2}},
\end{equation}
so that $\lim_{t \to \infty} \mathbb{E}\bigl[\widehat{B}_{it}(\delta)^{-2}\bigr] < \infty$. From (\ref{eq:c_convergence}) and the boundedness of $\widehat{B}_{it}(\delta)^{-1}$ and $\widehat{C}_{it}$, we also have:
\begin{align}
\label{eq:trB_bounded}
&\widehat{B}_{it}(\delta)^{-1}\widehat{C}_{it} \widehat{B}_{it}(\delta)^{-T} \xrightarrow{a.s.,L^2} \\
&\biggl(\sum_{j=1}^n \pi_j M_j + \delta I_{d_y}\biggr)^{-1} \biggl(\sum_{j=1}^n \pi_j M_j\biggr)\biggl(\sum_{j=1}^n \pi_j M_j + \delta I_{d_y}\biggr)^{-T}.\notag
\end{align}
Since $\widehat{H}_{it}$ and $\widehat{V}_{it}$ depend solely on $\hat{x}_i(t)$, they are independent of $v_i(t)$. Because $\mathbb{E}[v_j(\tau)]=0$, $\mathbb{E}[\widehat{H}_{j\tau}^T\widehat{V}_{j\tau}^{-1}v_j(\tau)]= 0$ and as before $\mathbb{E}[\widehat{\phi}_{it}]= 0$. Since $\widehat{B}_{it}(\delta)$ is independent of $v_{i}(t)$ as well, $\mathbb{E}\bigl[\widehat{B}_{it}(\delta)^{-2}\widehat{\phi}_{it} \bigr] = 0$ and a result equivalent to \eqref{eq:covariance-phiit} holds:
\begin{align}
  &\mathbb{E} [\widehat{B}_{it}(\delta)^{-1}\widehat{\phi}_{it} \widehat{\phi}_{it}^T\widehat{B}_{it}(\delta)^{-T}] \notag\\
  &= \mathbb{E} \biggl[ \widehat{B}_{it}(\delta)^{-1} \biggl(\frac{1}{(t+1)^2}\sum_{j=1}^n\sum_{\tau=0}^{t-1} \bigl[\mathcal{K}^{t-\tau-1}\bigr]_{ij}^2 \widehat{M}_{j\tau}\biggr) \widehat{B}_{it}(\delta)^{-T} \biggr] \notag\\
  &\preceq \frac{1}{t+1} \mathbb{E} \bigl[\widehat{B}_{it}(\delta)^{-1}\widehat{C}_{it}\widehat{B}_{it}(\delta)^{-T}\bigr].\label{eq:covariance-hatphiit}
\end{align}
Finally, we can consider the mean squared error:
\begin{align*}
&\mathbb{E}\bigl[\|\hat{y}_i(t) - y\|_2^2 \bigr]=\mathbb{E} \biggl\|\widehat{B}_{it}(\delta)^{-1}\frac{\omega_{it}}{t+1} - \widehat{B}_{it}(\delta)^{-1}\widehat{B}_{it}(\delta)y\biggr\|_2^2\\
&= \mathbb{E} \biggl\|\widehat{B}_{it}(\delta)^{-1} \biggl(g_{it} + \widehat{C}_{it}y + \widehat{\phi}_{it} - (G_{it} + \widehat{C}_{it}+ \delta I_{d_y})y \biggr)  \biggr\|_2^2\\
&= \mathbb{E} \|\widehat{B}_{it}(\delta)^{-1}(b_{it} +\widehat{\phi}_{it} +\delta y)\|_2^2\\
&= \mathbb{E}\biggl[b_{it}^T \widehat{B}_{it}(\delta)^{-2} b_{it} + \widehat{\phi}_{it}^T\widehat{B}_{it}(\delta)^{-2}\widehat{\phi}_{it} + \delta^2 y^T\widehat{B}_{it}(\delta)^{-2}y\\
&\quad + 2b_{it}^T\widehat{B}_{it}(\delta)^{-2}\widehat{\phi}_{it} + 2\delta y^T\widehat{B}_{it}(\delta)^{-2}\widehat{\phi}_{it}+ 2\delta b_{it}^T\widehat{B}_{it}(\delta)^{-2}y\biggr]\\
&= b_{it}^T \mathbb{E}\bigl[ \widehat{B}_{it}(\delta)^{-2}\bigr] b_{it} + \tr( \mathbb{E} \bigl[\widehat{B}_{it}(\delta)^{-1} \widehat{\phi}_{it} \widehat{\phi}_{it}^T \widehat{B}_{it}(\delta)^{-T}\bigr])\\
&\qquad + \delta^2y^T\mathbb{E}\bigl[ \widehat{B}_{it}(\delta)^{-2}\bigr]y + 2\delta b_{it}^T\mathbb{E}\bigl[\widehat{B}_{it}(\delta)^{-2}\bigr]y\\
&\longineq{\eqref{eq:covariance-hatphiit}}{}{\leq} b_{it}^T \mathbb{E}\widehat{B}_{it}(\delta)^{-2} b_{it} + 2\delta b_{it}^T\mathbb{E}\widehat{B}_{it}(\delta)^{-2}y + \delta^2y^T\mathbb{E}\widehat{B}_{it}(\delta)^{-2}y\\
&\qquad+ \frac{1}{t+1}\tr\biggl( \mathbb{E} \biggl[\widehat{B}_{it}(\delta)^{-1} \widehat{C}_{it} \widehat{B}_{it}(\delta)^{-T} \biggr]\biggr)\\
& \to \delta^2y^T\bigl(\sum_{j=1}^n \pi_j M_j + \delta I_{d_y}\bigr)^{-2}y.
\end{align*}
In the final step, the first two terms go to zero because $b_{it} \to 0$ and $\lim_t \mathbb{E}\widehat{B}_{it}(\delta)^{-2} < \infty$ from (\ref{eq:invB_bounded}), the third term converges in view of (\ref{eq:invB_bounded}) again, while the last term goes to zero because the trace is bounded in the limit in view of (\ref{eq:trB_bounded}).\hfill$\qed$